\documentclass{article}

\usepackage{subfigure}
\usepackage[disable]{todonotes}
\usepackage{amssymb}
\usepackage{graphicx}
\usepackage{mathtools}
\usepackage{todonotes}
\usepackage{amsthm}
\usepackage{url}

\usepackage{amsthm}
\newtheorem*{rep@theorem}{\rep@title}
\newcommand{\newreptheorem}[2]{%
\newenvironment{rep#1}[1]{%
 \def\rep@title{#2 \ref{##1}}%
 \begin{rep@theorem}}%
 {\end{rep@theorem}}}

\newtheorem{theorem}{Theorem}[section]

\newtheorem{lemma}[theorem]{Lemma}

\newreptheorem{theorem}{Theorem}
\newreptheorem{lemma}{Lemma}
\newreptheorem{proposition}{Proposition}
\newreptheorem{claim}{Claim}
\newreptheorem{fact}{Fact}
\usepackage[linesnumbered]{algorithm2e}

\author{Dan Alistarh \\ IST Austria \and
Nikita Koval \\ IST Austria, JetBrains \and
Giorgi Nadiradze \\ IST Austria}
\date{}

\title{Efficiency Guarantees for Parallel Incremental Algorithms under Relaxed Schedulers}

\usepackage[linesnumbered]{algorithm2e}
\SetKw{Continue}{continue}
\SetKw{Break}{break}

\newtheorem{claim}{Claim}
\DeclareMathOperator{\poly}{poly}

\begin{document}

\maketitle{}

\begin{abstract}
    Several classic problems in graph processing and computational geometry are solved via incremental algorithms, which split computation into a series of small tasks acting on shared state, which gets updated progressively. 
    While the sequential variant of such algorithms usually specifies a fixed (but sometimes random) order in which the tasks should be performed, a standard approach  to parallelizing such algorithms is to \emph{relax} this constraint to allow for out-of-order parallel execution. This is the case for parallel implementations of Dijkstra's single-source shortest-paths (SSSP) algorithm, and for parallel Delaunay mesh triangulation.
    While many software frameworks parallelize incremental computation in this way, it is still not well understood whether this relaxed ordering approach can still provide any complexity guarantees. 
    
    In this paper, we address this problem, and analyze the efficiency guarantees provided by a range of incremental algorithms when parallelized via relaxed schedulers. 
    We show that, for algorithms such as Delaunay mesh triangulation and sorting by insertion, schedulers with a maximum relaxation factor of $k$ in terms of the maximum priority inversion allowed will introduce a maximum amount of wasted work of $O(\log n \textnormal{ poly } (k) ), $ where $n$ is the number of tasks to be executed. 
    For SSSP, we show that the additional work is  $O( \textnormal{ poly } (k) \, d_{
    max} / w_{min}), $ where $d_{\max}$ is the maximum distance between two nodes, and $w_{
    \min}$ is the minimum such distance. In practical settings where $n \gg k$, this suggests that the overheads of relaxation will be outweighed by the improved scalability of the relaxed scheduler. 
    On the negative side, we provide lower bounds showing that certain algorithms will inherently incur a non-trivial amount of wasted work due to scheduler relaxation, even for relatively benign relaxed schedulers.  
\end{abstract}

\section{Introduction}

Several classic problems in graph processing and computational geometry can be solved  \emph{incrementally}: algorithms are structured as a series of \emph{tasks}, each of which examines a subset of the  algorithm state, performs some computation, and then updates the state. 
For instance, in Dijkstra's classic graph single-source shortest paths (SSSP) algorithm~\cite{dijkstra1959note}, the state consists of the current distance estimates for each node in the graph, and each task corresponds to a node ``relaxation," which may update the distance estimates of the node's neighbors. In the case of  the classic sequential variant, the order in which tasks get executed is dictated by the sequence of node distances. 
At the same time, many other incremental algorithms, such as Delaunay mesh triangulation, assume arbitrary (or random) orders on the tasks to be executed, and can even provide efficiency guarantees under such orderings~\cite{blelloch2016parallelism}. 

A significant amount of attention has been given to 
\emph{parallelizing} such incremental iterative algorithms, e.g.~\cite{gonzalez2012powergraph, Nguyen13, blelloch2016parallelism,  dhulipala17julienne, dhulipala2018theoretically}. 
One approach has been to study the dependence structure of such algorithms, proving that, in many cases, the dependency chains are \emph{shallow}. 
This can be intuitively interpreted as proving that such algorithms should have significant levels of parallelism. 
One way to exploit this fine-grained parallelism, e.g.~\cite{BFS12, shun13priority} has been to carefully split the execution into task prefixes of limited length, and to parallelize the execution of each prefix efficiently. While this approach can be  efficient, it does require an understanding of the workload and task structure, and may not be immediately applicable to algorithms where the task ordering is dependent on the input. 

An alternative approach has been to employ scalable data structures with only ensure \emph{relaxed priority order} to schedule task-based programs.
The idea can be traced back to Karp and Zhang~\cite{KarZha93}, who studied parallel backtracking in the PRAM model, and noticed that, in some cases, the scheduler can relax the strict order induced by the sequential algorithm, allowing tasks to be processed speculatively ahead of their dependencies, without loss of correctness. 
For instance, when parallelizing SSSP, e.g.~\cite{SprayList, making-kjell-happy, Nguyen13}, the scheduler may retrieve vertices in arbitrary order without breaking correctness, as the distance at each vertex is guaranteed to eventually converge to the minimum.  
However, there is intuitively a trade-off between the performance gains arising from using scalable relaxed schedulers, and the loss of determinism and the possible wasted work due to having to re-execute speculative tasks. 

This approach is quite popular in practice, as several efficient relaxed schedulers have been proposed~\cite{LotanShavit, Basin11, klsm, SprayList, Haas, Nguyen13, MQ, AKLN17, sagonas2017contention}, which can attain state-of-the-art results for graph processing and machine learning~\cite{Nguyen13, gonzalez2012powergraph}, and even have hardware implementations~\cite{Swarm}. At the same time, despite showing good empirical performance, this approach does not come with analytical bounds: in particular, for most known algorithms, it is not clear how the relaxation factor in the scheduler affects the total work performed by the parallel algorithm. 

We address this question in this paper. Roughly, we show under analytic assumptions that, for a set of fundamental algorithms including parallel Dijkstra's and Delaunay mesh triangulation, the extra work engendered by scheduler relaxation can be negligible with respect to the total number of tasks executed by the sequential algorithm. On the negative side, we show that relaxation does not come for free: we can construct worst-case instances where the cost of relaxation is asymptotically non-negligible, even for relatively benign relaxed schedulers. 

We model the relaxed execution of incremental algorithms as follows. The \emph{algorithm} is specified as an ordered sequence of tasks, which may or may not have precedence constraints. 
The algorithm's execution is modeled as an interaction between a \emph{processor}, which can execute tasks, modify state, and possibly create new tasks, and a  \emph{scheduler}, which stores the tasks in a priority order specified by the algorithm. 
At each step, the processor requests a new task from the scheduler, examines whether the task can be processed (i.e., that all precedence constraints are satisfied), and then executes the task, possibly modifying the state and inserting new tasks as a consequence. 

An exact scheduler would return tasks following priority order. Since ensuring such strict order semantics is known to lead to contention and poor performance~\cite{Alistarh14}, practical scalable schedulers often \emph{relax} the priority order in which tasks are returned, up to some constraints. 
For generality, in this paper, we assume when proving performance upper bounds that the scheduler may in fact be \emph{adversarial}---actively trying to get the algorithm to waste steps, up to some natural \emph{rank inversion} and \emph{fairness} constraints. Specifically, the two natural constraints we enforce on the scheduler are on 1) the \emph{maximum rank inversion} between the highest priority task present and the rank of the task returned, and on 2) \emph{fairness}, in terms of the maximum number of schedule steps for which the task of highest priority may remain unscheduled. For convenience, we upper bound both these quantities by a parameter $k$, which we call the \emph{relaxation factor} of the scheduler. 
Simply put, a $k$-relaxed scheduler must 1) return one of the $k$ highest-priority elements in every step; and 2) return a task at the latest $k$ steps after it has become the highest-priority task available to the scheduler. 
We note that real schedulers enforce such constraints either deterministically~\cite{klsm} or with high probability~\cite{AKLN17, alistarhbkln18, alistarhBKN18}.

A significant limitation of the above model is that it is \emph{sequential}, as it assumes a single processor which may execute tasks. While our results will be developed in this simplified sequential model, we also discuss a parallel version of the model in Section~\ref{sec:parallel-model}. 

It is natural to ask whether incremental algorithms can still provide any guarantess on total work performed under $k-$relaxed schedulers. 
Additional work may arise due to relaxation for two reasons.  
The first is if the parallel execution enforces ordering constraints between data-dependent tasks: for instance, when executing a graph algorithm, the task corresponding to a node $u$ may need to be processed \emph{before} the task corresponding to any neighbor which has higher priority in the initial node ordering. 
A second cause for wasted work is if a task may need to be re-executed once the state is updated: this is the case when running parallel SSSP: due to relaxation, a node may be speculatively relaxed at a distance that is higher than its optimal distance from the source, leading it to be relaxed several times.  
We note that neither phenomenon may occur when the priority order is strict---since the top priority task cannot have preceding constraints nor need to be re-executed---but are inherent in parallel executions. 

A trivial upper bound on wasted work for an algorithm with total work $W$ under a $k$-relaxed scheduler would be $O(kW)$---intuitively, in the worst case the scheduler may return $k$ tasks before the top priority one, which can always be executed without violating constraints. 
The key observation we make in this work is that, because of their local dependency structure, some popular incremental algorithms will incur \emph{significantly less} overhead due to out-of-order execution. 

More precisely, for incremental algorithms, such as Delaunay mesh triangulation and sorting by insertion, we show that the expected overhead of execution via a $k$-relaxed scheduler is \\ $O(\poly{k} \log n) $, where $n$ is the number of tasks the sequential variant of the algorithm executes. 
We exploit the following properties of incremental algorithms, shown in \cite{blelloch2016parallelism}:
The probability that the task at position $j$ is dependent on the task at position $i<j$ depends only on the tasks at positions $1,2,...,i$ and $j$, and assuming a random order of tasks, this probability is upper bounded by $O(1/i)$.
While the technical details are not immediate, the argument boils down to bounding, for each top-priority task, the number of \emph{dependent} tasks which may be returned by the scheduler while the task is still in the scheduler queue. 

For  SSSP, which does not have a dependency structure but may have to re-execute tasks, we use a slightly different approach, based on $\Delta$-stepping, ~\cite{meyer2003delta}. We bound the total overhead of relaxation to $O( \poly {k} \, d_{\max} / w_{\min} )$, where $d_{\max}$ is the maximum distance between two nodes, and $w_{
    \min}$ is the minimum such distance. While this overhead may in theory be arbitrarily large, depending on the input, we note that for many graph models, this overhead is small. (For example, for Erdös-Renyi graphs with constant weights, the overhead is $O( \poly {k} \log n)$.)

It is interesting to interpret these overheads in the context of practical concurrent schedulers such as MultiQueues~\cite{MQ, gonzalez2012powergraph}, where the relaxation factor $k$ is proportional to the number of concurrent processors $p$, up to logarithmic factors. 
Since in most instances the size of the number of tasks $n$ is significantly larger than the number of available processors $p$, the overhead of relaxation can be seen to be comparatively small. This observation has been already made empirically for specific instances, e.g.~\cite{lenharth2015priority}; our analysis formalizes this observation in our model.

On the negative side, we also show that the overhead of relaxation is non-negligible in some instances. Specifically, we exhibit instances of incremental algorithms
where the overhead of relaxed execution is $\Omega(\log n)$. 
Interestingly, this lower bound does not require the scheduler to be adversarial: we show that it holds even in the case of the relatively benign MultiQueue scheduler~\cite{MQ, AKLN17}. 

\paragraph{Related Work.} 
Parallel scheduling of iterative algorithms is a vast area, so a complete survey is beyond our scope. We begin by noting that our results are not relevant to standard work-stealing schedulers~\cite{blumofe, Cilk} since such schedulers do not provide any guarantees in terms of the \emph{rank of elements removed}.\footnote{We are aware of only one previous attempt to add priorities to work-stealing schedulers~\cite{imam2015load}, using a multi-level global queue of tasks, partitioned by priority. This technique is different, and provides no work guarantees.}

An early instance a relaxed scheduler is in the work of Karp and Zhang~\cite{KarZha93}, for parallelizing backtracking in the PRAM model.  
This area has recently become very active, and several relaxed scheduler designs have been proposed, trading off relaxation and scalability, e.g.~\cite{LotanShavit, Basin11, klsm, SprayList, Haas, Nguyen13, MQ, AKLN17, sagonas2017contention}. 
In particular, software packages for graph processing~\cite{Nguyen13} and machine learning~\cite{gonzalez2012powergraph} implement such relaxed schedulers. 

Our work is  related to the line of research by Blelloch et al.~\cite{Blelloch, BFS12, blelloch2012internally, BFS14, blelloch2016parallelism}, as well as~\cite{coppersmith1987parallel, calkin1990probabilistic, FN18}, 
which examines the dependency structure of a broad set of iterative/incremental algorithms and exploit their inherent parallelism for fast implementations.  We benefit significantly from the analytic techniques introduced in this work. 

We note however some important differences between these results and our work. 
The first difference concerns the scheduling model: references such as~\cite{BFS12, BFS14,blelloch2016parallelism} assume a synchronous PRAM execution model, and focus on analyzing the maximum dependency length of algorithms under random task ordering, validating the results via concurrent implementations. 
By contrast, we employ a relaxed scheduling model, that models data processing systems based on relaxed priority schedulers, such as~\cite{Nguyen13}, and provide work bounds for such executions.
Although superficially similar, our analysis techniques are different from those of e.g.~\cite{BFS12, blelloch2016parallelism} since our focus is not on the maximum dependency depth of the algorithms, but rather on the number of local dependencies which may be exploited by the adversarial scheduler to cause wasted work. 
We emphasize that the fact that the algorithms we consider may have low dependency depth does not necessarily help, since a sequential algorithm could have low dependency depth and be inefficiently executable by a relaxed scheduler: a standard example is when the dependency depth is low (logarithmic), but each ``level" in a breadth-first traversal of the dependency graph has high fanout. This has low depth, but would lead to high speculative overheads. 
(In practice, greedy graph coloring on a dense graph would lead to such an example.) 

A second difference concerns the interaction between the scheduler and the algorithm. The scheduling mechanisms proposed in e.g.~\cite{BFS12} assume knowledge about the algorithm structure, and in particular adapt the length of the prefix of tasks which can be scheduled at any given time to the structure of the algorithm. 
In contrast, we assume a basic scheduling model, which may even be adversarial (up to constraints), and show that such schedulers, which relax priority order for increased scalability, inherently provide bounded overheads in terms of wasted work due to relaxation. 

Finally, we note that references such as~\cite{BFS12, blelloch2016parallelism} focus on algorithms which are efficient under  random orderings of the tasks. In the case of SSSP, we show that relaxed schedulers can efficiently execute algorithms which have a fixed optimal ordering. 

Another related reference is~\cite{alistarhBKN18}, in which we introduced the scheduling model used in this paper, related it to MultiQueue schedulers~\cite{MQ}, and analyzed the work complexity of some simple iterative greedy algorithms such as maximal independent set or greedy graph coloring. We note the technique introduced in this previous paper only covered a relatively limited set of iterative algorithms, where the set of tasks are defined and fixed in advance, and focused on the complexity of greedy maximal independent set (MIS) under relaxed scheduling. In contrast, here we consider more complex \emph{incremental} algorithms, in which tasks can be added and modified dynamically. Moreover, as stated, here we also cover algorithms such as SSSP, in which computation should follow a fixed sequential task ordering, as opposed to a random ordering which was the case for greedy MIS.

\section{Relaxed Schedulers: The Sequential Model}

We begin by formally introducing our sequential model of relaxed priority schedulers.   
We represent a priority scheduler as a \emph{relaxed ordered set} data structure $Q_k$, where the integer parameter $k$ is the \emph{relaxation factor}.
A relaxed priority scheduler contains $<task,priority>$ pairs and supports the following operations:

\begin{enumerate}
    \item $Q_k.Empty()$, returns \texttt{true} if $Q_k$ is empty, and \texttt{false} otherwise.
    \item $ApproxGetMin()$, returns a $<task,priority>$ pair if one is available, without deleting it.
    \item $DeleteTask(task)$, removes specified task from the scheduler. This is used to remove a task returned by $ApproxGetMin()$, if applicable. 
    \item $Insert(<task, priority>)$, inserts a new task-priority pair in $Q_k$.
\end{enumerate}

We denote the rank (in $Q_k$) of the task returned by the $t$-th $ApproxGetMin()$ operation by $rank(t)$, and call it the rank of a task returned on step $t$. 
For a task $u$, let $inv(u)$ be the number of inversions experienced by task $u$ between the step when $u$ becomes the highest priority task in $Q_k$ and the step when task $u$ is returned by the scheduler.
That is, $inv(u)+1$ is the number of $ApproxGetMin()$ operations needed for the highest priority task $u$ to be scheduled.

\paragraph{Rank and Fairness Properties.} 
The relaxed priority schedulers $Q_k$ we consider will enforce the following properties:

\begin{enumerate}
    \item $RankBound$: at any time step $t$, $rank(t) \le k$.
    \item $Fairness$: for any task $u$, $inv(u) \le k-1$. 
\end{enumerate}

Priority schedulers such as k-LSM~\cite{klsm} enforce these properties deterministically, where $k$ is a tunable parameter. 
We have shown in previous work that the MultiQueue~\cite{MQ} scheduler ensures these properties both in sequential and concurrent executions~\cite{AKLN17, alistarhbkln18} with parameter $k = O( q \log q )$, with exponentially low failure probability in $q$, the number of queues.

Next, we describe how incremental algorithms can be implemented in this context.

\section{Incremental Algorithms} 

\subsection{General Definitions}

We assume a model of incremental algorithms which execute a set of tasks iteratively, one by one, and where each task incrementally  updates the algorithm's state. 
For example, in incremental graph algorithms, 
the shared state corresponds to a data structure storing the graph nodes, edges, and meta-data corresponding to nodes. 
Tasks usually correspond to vertex operations, and are usually inserted and executed in some order, given by the input. 
If this task order is random, we say that the incremental algorithm is randomized. We will consider both randomized incremental algorithms, where each task has a priority based on the random order, and deterministic ones, where the order is fixed. Using an exact scheduler corresponds to executing tasks in the same order as the sequential algorithm, while using a relaxed scheduler allows
out-of-order execution of tasks. 

\paragraph{Definition.} More formally, randomized incremental algorithms such as Delaunay triangulation and comparison sorting with via BST insertion can be modelled as follows:

We are given $n$ tasks, which must be executed iteratively in some (possibly random) order.
Initially, each task $u$ is assigned a unique label $\ell(u)$. For instance, this label can be based on a random permutation of $n$ given tasks,  $\pi$. 
That is, for task $u$, $\ell(u)=i$, iff $\pi(i)=u$. 
A lower label can be equated with higher priority.
Each task performs some computation and updates the algorithm state. In the case of Delaunay triangulation, tasks update the triangle mesh, while in the case of Comparison Sorting tasks modify the BST accordingly.
Generic sequential pseudocode is given in Algorithm \ref{IncrementalSequential}. We note that a similar generic algorithm was presented in~\cite{alistarhBKN18} for parallelizing greedy iterative algorithms. 

\begin{algorithm} 
\caption{General Framework for incremental algorithms. }
\label{IncrementalSequential}
    \KwData{Sequence of tasks $V=(v_1, v_2, ..., v_n)$, in decreasing priority order. }
    \CommentSty{// $Q$ is an exact priority queue.} \\    
    $Q \gets \text{tasks of $V$ with priorities}$ \\
    \For{each step $t$}
    {
        \CommentSty{// remove the task with highest priority. } \\
        $v_t \gets Q.DeleteMin()$  \\  
        $Process(v_t)$ \\
        \CommentSty{// stop if $Q$ is empty. } \\
        \If{$Q.empty()$} { \Break }
    }
\end{algorithm}

When using a relaxed priority $Q_k$ instead of an exact priority queue $Q$, one issue is the presence of inter-task dependencies. 
These dependencies are specified by the algorithm, and are affected by the permutation of the tasks:
For comparison sorting, a task depends on all of its ancestor tasks in the resulting BST, while  
for Delaunay Triangulation there is a dependency between two tasks if right before either one is added,
their encroaching regions overlap by at least an edge in the mesh. (Due to space constraints, we will assume the reader is familiar with terminology related to Delaunay mesh triangulation. We direct the reader to e.g.~\cite{blelloch2016parallelism} for an overview of sequential and parallel algorithms for this problem.) 

If task $v$ depends on task $u$ and $\ell(u) < \ell(v)$, then task $v$ can not be processed before  task $u$. 
We call task $u$ an \emph{ancestor} of task $v$ in this case.
We assume that the task returned by the relaxed scheduler can be processed only if all of its ancestors are 
already processed.
Pseudocode is given in Algorithm \ref{IncrementalRelaxed}.

\begin{algorithm} 
\caption{General Framework for executing incremental algorithms using relaxed priority schedulers. }
\label{IncrementalRelaxed}
    \KwData{Sequence of tasks $V=(v_1, v_2, ..., v_n)$, in decreasing priority order. \\ }
    \CommentSty{// $Q_k$ is a relaxed priority queue.} \\    
    $Q_k \gets \text{tasks of $V$ with priorities}$ 
    \For{each step $t$}
    {
        \CommentSty{// get the task with highest priority from $Q_k$. } \\
        $v_t \gets Q_k.GetMin()$  \\  
        \CommentSty{// check if $v_t$ has no dependencies. } \\
        \If{$CheckDependencies(v_t)$} {
            $Q_k.Delete(v_t)$ \\
            $Process(v_t)$ \\
        }
        \CommentSty{// stop if $Q$ is empty. } \\
        \If{$Q_k.empty()$} { \Break }
    }
\end{algorithm}

Observe that the $For$ loop runs for exactly $n$ steps in the exact case, but it may require extra steps in the relaxed case. 
We are interested in upper bounding the number of extra steps, since this is a measure of the additional work incurred when executing via the relaxed priority queue.
In order to do this, we need to specify some properties for the dependencies of the incremental algorithms we consider.

Denote by $p_{ij}$ be probability that task with label $j$ depends on task with label $i$.
We require the incremental algorithms to have the following properties: 
\begin{enumerate}
    \item for each pair of task indices $i<j$, $p_{ij} \le C/i$, where $C$ is large enough constant which depends on the incremental algorithm.
    \item for each pair $i<j$, $p_{ij}$ depends only on tasks with labels $1,..., i$ and $j$.
\end{enumerate}

The fact that comparison sorting and Delaunay triangulation have the above properties has already been shown in~\cite{blelloch2016parallelism}. More precisely, for comparison sorting, these properties are proved in \cite[Section 3]{blelloch2016parallelism}.
In the case of Delaunay triangulation, property (2) is showed in the same paper~\cite[Section 4]{blelloch2016parallelism}, while property (1) follows from~\cite[Theorem 4.2]{blelloch2016parallelism}.

\subsection{Analysis}

In this subsection, we prove an upper bound on the number of extra steps required by our generic relaxed framework for executing incremental algorithms. As a first step, we will derive some additional properties of the relaxed scheduler. 

Let $A_{ij}$ be the event that task with label at least $j$ is returned by the scheduler before task
with label $i$ is processed by the incremental algorithm. 
Observe that if the scheduler returns the highest priority task, then this task can always be processed by the incremental algorithm, since this task is guaranteed to have no ancestors.
\begin{lemma} \label{lem:Aij}
If $j-i \ge 2k^2$, $Pr[A_{ij}]=0.$ 
\end{lemma}
\begin{proof}
For labels $j-i \ge 2k^2$,
let $u$ be the task with label $i$ and let $v$ be some task with label at least $j$.
Also, let $t$ be the earliest step at which rank of $u$ is at most $k$.
This means that at time step $t-1$, $rank(u) > k$ and by rank property no tasks with labels larger than $i$ were scheduled at time steps $1,2,...,t-1$.
Thus, we have that at time step $t$, $rank(v) > j-i \ge 2k^2$.
Because of the fairness property it takes $k$ steps to remove the task with highest priority(lowest label), so task $u$ will be returned by the scheduler and subsequently will be processed by the algorithm no later than at time step $t+k^2$.
Rank of $v$ can decrease by at most 1 after each step, thus at time step $t+k^2$,
$rank(v) > 2k^2-k^2 \ge k$. Hence, $v$ can be returned by the scheduler only after time step $t+k^2$ and this gives us that $Pr[A_{ij}]=0.$
\end{proof}
For any label $i$, let $R_{i}$ be the number of times scheduler returns 
task with label greater than $i$ (some task can be counted multiple times), before task with label $i$ is processed by the algorithm. The following holds:
\begin{lemma} \label{lem:Rij}
For any $i$, $R_{i} \le k^2.$
\end{lemma}
\begin{proof}
Let $u$ be the task with label $i$.
Also, let $t$ be the earliest step at which rank of $u$ is at most $k$.

This means that at time step $t-1$, $rank(u) > k$ and by rank property no task with label at least $i$ can be returned by the scheduler at time steps $1,2,...,t-1$.
Because of the fairness property it takes $k$ steps to remove the task with highest priority(lowest label), so task $u$ will be returned by the scheduler and subsequently will be processed by the algorithm no later than at time step $t+k^2$.
Trivially, the total number of times some task with label at least $i$(or in fact any label) can be returned by the scheduler before the time step $t+k^2$ is $k^2$.
\end{proof}
With the above lemmas in place we can proceed to prove an upper bound for the extra number of steps.

\begin{theorem} \label{SeqTheorem}
The expected number of extra steps is upper bounded by $O(poly(k)\log{n})$.
\end{theorem}

\begin{proof}
Let $D_{ij}$ be the event that task with label $j$ depends on task with label $i<j$.
From the properties of incremental algorithms we consider, we get that $Pr[D_{ij}]=p_{ij}\le C/i$.

Recall that for $i<j$, $A_{ij}$ is the event that task with label at least $j$ is returned by relaxed scheduler before task with label $i$ is processed by the algorithm. Observe that at $A_{ij}$ and  $D_{ij}$ are independent.
Since, $D_{ij}$ depends only on the initial priorities of the tasks and does not depend on the relaxed scheduler.
On the other hand, it is easy to see that $Pr[A_{ij}|D_{ij}]=Pr[A_{ij}|\neg D_{ij}].$

Every extra step is caused by a task with an ancestor which is not processed.
Let $v$ be the label of the task we are not able to process because of dependencies and let $u$ be the highest priority ancestor task of $v$. If $u$ also has an unprocessed ancestor , we repeat the same step.
Eventually we can recurse to the pair of tasks $(u',u'')$ such that $u'$ is highest priority ancestor of $u''$ and all ancestors of $u'$ are already processed.
Let $\ell(u')=i$ and $\ell(u'')=j$, we charge the extra step to the pair of labels $(i,j)$.

Note that pair of labels $(i,j)$ can be charged only if $D_{ij}$ and $A_{ij}$.
Let $L_{ij}$ be the event that $(i,j)$ is charged at least once. That is, $L_{ij}$ will happen if and only if $D_{ij}$ and $A_{ij}$ happen.
Also it is easy to see that the total number of times $(i,j)$ can be charged is upper bounded by $R_{i}$(Recall \ref{lem:Rij}).

\begin{align}  \label{TotalCharges}
E[\# extra steps] &\le \sum_{i=1}^{n-1} \sum_{j=i+1}^{n} Pr[L_{ij}]R_{i} \nonumber \\
&\underset{\le}{\text{Lemma }\ref{lem:Rij}} \sum_{i=1}^{n-1} \sum_{j=i+1}^{n} Pr[D_{ij}]Pr[A_{ij}]k^2 \nonumber \\ & \le
\sum_{i=1}^{n-1} \sum_{j=i+1}^{n} \frac{C}{i}Pr[A_{ij}]k^2 \nonumber \\ &\le
\sum_{i=1}^{n-1} \sum_{j=i+1}^{i+2k^2} \frac{C}{i}Pr[A_{ij}]k^2 \nonumber  +
\sum_{i=1}^{n-1} \sum_{j=i+2k^2+1}^{n} \frac{C}{i}Pr[A_{ij}]k^2 \nonumber \\ 
&\underset{=}{\text{Lemma }\ref{lem:Aij}} 
\sum_{i=1}^{n-1} \sum_{j=i+1}^{i+2k^2} \frac{C}{i}Pr[A_{ij}]k^2 \nonumber \\ &\le
\sum_{i=1}^{n-1}\frac{C}{i} 2k^4 \le O(k^4 \log{n}). 
\end{align}
\end{proof}

\section{RELAXED SCHEDULERS: THE TRANSACTIONAL MODEL} \label{sec:parallel-model}

We now consider an alternative model where tasks are executed \emph{concurrently}, each as part of a (software or hardware) transaction. This is unlike our standard model, which is entirely sequential. It is important to note that the correspondence between the two models is not one-to-one, since in this concurrent model transactions may abort due to data conflicts. 
More precisely, we assume that the algorithm consists of $n$ tasks, each corresponding to some transaction.
Transactions are scheduled by an entity called the the \emph{transactional scheduler}.
Every task $u$ has label $\ell(u)$, where a lower label corresponds to higher priority.
In transactional model, unlike sequential model, we assume that transaction aborts if and only if it is executed concurrently with a transaction it depends on.
In other words, dependencies create data conflicts for concurrent transactions and conflicts are resolved 
in favor of higher priority transaction.
Another crucial difference is that in transactional model we assume an upper bound
on the interval contention. That is, each transaction can be concurrent with at most $C$ transactions 
in total(during one execution). This is needed because, If $u$ is the transaction with highest priority and $v$ is the transaction with second highest priority, which depends on $u$, then $u$ can cause 
$v$ to be aborted large number of times, even in the case of exact scheduler.

\paragraph{Properties of the transactional scheduler} 
For transaction $u$, let $inv(u)$ be the number of transactions returned by \emph{transactional scheduler} after the point
$u$ becomes the highest priority transaction available to the scheduler and before it is returned by the scheduler.
We require \emph{transactional scheduler} to have the following properties, which are similar to
the properties in sequential model:
\begin{enumerate}
    \item \emph{RankBound}: transaction $u$ with label $\ell(u)$ is available to the \emph{transactional scheduler}
    only after at least $\ell(u)-k$ transactions with higher priority than $u$ are executed successfully.
    \item \emph{Fairness}: For any transaction $u$, $inv(u) \le k-1$.
\end{enumerate}

Next, we derive concurrent versions of lemmas proved in the sequential setting.
Let $A_{ij}$ be the event that transaction with label at least $j$ is executed concurrently with the transaction with label $i<j$ or returned by the \emph{transactional scheduler} before the transaction with label $i$.
Observe that if scheduler returns highest priority transaction, then this transaction will never abort.

\begin{lemma}
If $j-i\ge 2k(C+k)$, $Pr[A_{ij}]=0$.
\end{lemma}
\begin{proof}
Let $u$ be the transaction with label $i$ and let $v$ be a transaction with label at least $j$.
Consider first point when $u$ is available to the scheduler.
Observe that at this point, no transactions with label greater than
$i$ are available to the scheduler and by the $RankBound$ property,
there are at most $k-1$ transactions with higher priority than $u$ which are left to be processed.
By $Fairness$ property, there can be only $k-1$ transactions scheduled before the transaction with highest priority. Once the highest priority transaction is scheduled, there can be at most $C$ transactions
executing concurrently with it.
Thus, the total number of transactions which were running at some point during period between $u$ became available to the scheduler and was executed successfully is at most $k(C+k)$.
We get that at the point $u$ has finished successful execution, $v$ is not available to the scheduler, since the total number of successful transactions
is at most $\ell(u)+k(C+k) < \ell(v)-k$. Thus, $Pr[A_{ij}=0]$.
\end{proof}

Let $u$ be the task with label $i$. Let $R_i$ be the total number of times scheduler returns transaction with label greater than $i$ before it returns the transaction $u$, plus the number of transaction which are concurrent with $u$ at some point.

\begin{lemma}
For any transaction $u$ with $\ell(u)=i$, $R_i \le k(k+C)$.
\end{lemma}
\begin{proof}
As in the proof of the previous lemma, we can show that the total number of transactions which were running at some point during period between $u$ became available to the scheduler and was executed successfully is at most $k(C+k)$,  this trivially gives us that $R_i \le k(k+C)$.
\end{proof}

Now, we are ready to prove the following theorem:

\begin{theorem} 
The expected number of transactions aborted by an incremental algorithm is at most $O(k^2(C+k)^2\log{n})$
\end{theorem}

\begin{proof}
Let $D_{ij}$ be event that the transaction(task) with label $j$ depends on the transaction(task) with label $i<j$.
Note that transaction $i$ can abort transaction $j$ if and only if $A_{ij}$ and $D_{ij}$.
In transactional model, we charge the aborted transaction to the transaction which caused abort.
Each transaction can be charged at most $R_i \le k(k+C)$ times.
Observe that $R_i$ is a loose upper bound on the number of times transaction can be charged,
since charge to transaction $u$ can be caused by a concurrent transaction only.
With these properties in place, we can follow exactly the same steps as in the proof of Theorem \ref{SeqTheorem} to show that 
The expected number of transactions aborted by an incremental algorithm is at most $\sum_{i=1}^{n-1} \sum_{j=i+1}^{n} Pr[A_{ij} \text{ and } D_{ij}]R_{i}=O(k^2(C+k)^2\log{n})$.
\end{proof}

\section{Lower Bound on Wasted Work}

In this section, we prove the lower bound on the cost of relaxation in terms of additional work. 
We emphasize the fact that this argument does not require the scheduler to be adversarial: in fact, we will prove that a fairly benign relaxed priority scheduler, the MultiQueue~\cite{MQ}, can cause incremental algorithms to incur $\Omega( \log n) $ wasted work.

More precisely, let $inv_{i,i+1}$ be event that the relaxed scheduler returns the task with label $i+1$ before task with label $i$. First, we will prove the following claim for $MultiQueue$ being used as a relaxed scheduler:
\begin{claim}
\label{ConsecutiveLabelInversion}
For every $i>1$, $Pr[inv_{i,i+1}] \ge 1/8$.
\end{claim}
\begin{proof}
First we describe how incremental algorithms work using $MultiQueues$.
The $MultiQueue$ maintains $k$ sequential priority queues, where $k$ can be assumed to be a fixed parameter.
As before, each task is assigned a label according to the random permutation of input tasks (lower label means higher priority).
Initially, all tasks are inserted in $MultiQueue$ as follows:
for each task, we select one priority queue uniformly at random out, and insert the task into it.
To retrieve a task, the processor selects two priority queues uniformly at random and returns the task with highest priority (lowest label), among the tasks on the top of selected priority queues.

Let $\ell(u)=i$ and $\ell(v)=i+1$. Additionally, let  $q_u$ and $q_v$ be the queues where $u$ and $v$ are inserted in initially.
Also, let $T_{u,v}$ be event that $u$ and $v$ are the top tasks of queues at some point during the run of our algorithm.
We have that:
\begin{align}
Pr[inv_{i,i+1}] &=Pr[q_u \neq q_v]\Bigg(Pr[T_{u,v}]Pr[inv_{i,i+1}|T_{u,v},q_u \neq q_v] \\&+
Pr[\neg T_{u,v}]Pr[inv_{i,i+1}|\neg T_{u,v},q_u \neq q_v] \Bigg) \\
&\ge (1-\frac{1}{k})Min\Bigg(Pr[inv_{i,i+1}|T_{u,v},q_u \neq q_v],  \\ & \qquad \qquad \qquad Pr[inv_{i,i+1}|\neg T_{u,v},q_u \neq q_v]\Bigg).
\end{align}
Observe that if $\neg T_{u,v}$ and $q_u \neq q_v$, this means that tasks $u$ and $v$ are never compared against each other. Consider two runs of our algorithm until it returns either $u$ or $v$, first with initially chosen $q_u$ and $q_v$ and second with $q_u$ and $q_v$ swapped (these cases have equal probability of occurring).
Since vertices $u$ and $v$ have consecutive labels and are never compared by the MultiQueue,
this means that all the comparison results are the same in both cases, hence the scheduler 
has equal probability of returning $u$ or $v$. (It is worth mentioning here is that $T_{u,v}$ only depends on values $q_u$ and $q_v$ and does not depend on their ordering.)

This means that :
\begin{equation}
Pr[inv_{i,i+1}|\neg T_{u,v},q_u \neq q_v] \ge 1/2.    
\end{equation}
Now we look at the case where $u$ and $v$ are top tasks of queues at some step $t$.
Let $X_u$ be event that $u$ is returned by MultiQueue and similarly, let $X_v$ be event that $v$ is returned. We need to lower bound the probability that $X_v$ happens before $X_u$ .
We can safely ignore all the other tasks returned by scheduler and processed by algorithm since it is independent of whether $u$ or $v$ is returned first.
Let $r$ be the number of top tasks in queues which have labels larger than $i+1$.
At step $t$, $Pr[X_u]=3/k^2+2r/k^2$ and $Pr[X_v]=1/k^2+2r/k^2$, So we have that 
\begin{equation} 
Pr[X_u] \le 3 Pr[X_v].
\end{equation}
Observe that during the run of algorithm  $r$ will start to increase but we will always have invariant that $Pr[X_u] \le 3 Pr[X_v]$. This means that probability that $X_v$ happens before $X_u$ is at least:
\begin{equation}
    \frac{Pr[X_v]}{Pr[X_u]+Pr[X_v]} \ge 1/4.
\end{equation}
This gives us that:
\begin{equation}
Pr[inv_{i,i+1}| T_{u,v},q_u \neq q_v] \ge 1/4.    
\end{equation}
and consequently, since $1-1/k \ge 1/2$ we get that:
\begin{equation}
    Pr[inv_{i,i+1}] \ge (1-1/k)/4 \ge 1/8.
\end{equation}
\end{proof}
\begin{theorem}
For Delaunay triangulation and comparison sorting, the expected number of extra steps is lower bounded by $\Omega(\log{n})$.
\end{theorem}
\begin{proof}
To establish the lower bound, we can assume that if the scheduler returns vertex $v$,
which depends on some other unprocessed vertex, we check if vertex $u$ with label $\ell(v)-1$ is not processed and we charge edge $e=(u,v)$, if $v$ depends on $u$.
This way, we get that $p_{i,i+1}$ and $Pr[inv_{i,i+1}]$ are not correlated, since if we run algorithm to the point where vertex with label $i$ or $i+1$ is returned,
it will never check the dependency between them.

We will employ the following property of Delaunay triangulation and $BST$-based comparison sorting:
for any $i>0$, $p_{i,i+1} \ge 1/i$. This property is easy to  verify:  in Delaunay triangulation there is at least $1/i$ probability that vertices with labels $i$ and $i+1$ are neighbours in the Delaunay triangulation of vertices with labels $1,2,...,i,i+1$, in $BST$ based comparison sorting there is at least $1/i$ probability that
tasks with labels $i$ and $i+1$ have consecutive keys among keys of tasks with labels $1,2,...,i,i+1$
and in both cases the task with label $i+1$ will depend on the task with label $i$(see \cite{blelloch2016parallelism}).

This, in combination with Claim \ref{ConsecutiveLabelInversion}
will give us the lower bound on the number of extra steps, since if task with label $i+1$ depends on the task with label $i$ and it is returned first by scheduler, this will trigger at least one extra step, caused by not being able to process task:
\begin{equation}
    E[\# extra steps] \ge \sum_{i=1}^{n-1} p_{i,i+1}Pr[inv_{i,i+1}] \ge 1/8 \log{n}.
\end{equation}

\end{proof}

\section{Analyzing SSSP under Relaxed Scheduling}\label{section:sssp}

\paragraph{Preliminaries.}
Since the algorithm is different from the ones we considered thus far, we re-introduce some notation. 
We assume we are given a directed graph $G=(V,E)$ with positive edge weights $w(e)$ for each edge $e \in E$, and a source vertex $s$.
For each vertex $v \in V$, let $d(v)$ be the weight of a shortest path from $s$ to $v$.
Additionally, let $d_{max}=\max \{d(v) : v \in V\}$ and $w_{min}=\min \{ w(e) : e \in E \}$.

We consider the sequential pseudocode from Algorithm\ref{alg:1}, which uses a relaxed priority queue $Q_k$ to find shortest paths
from $s$ via a procedure similar to the $\Delta$-stepping algorithm~\cite{meyer2003delta}. 

In this algorithm $Q_k.push(v, dist)$ inserts a vertex $v$ with distance $dist$ in $Q_k$, $Q_k.pop()$ removes and returns a vertex, distance pair $(v,dist)$, such that $v$ is among the $k$ smallest distance vertices in $Q_k$. 
We also assume that $Q_k$ supports a $Q_k.DecreaseKey(v, dist)$ operation, which atomically decreases the distance of vertex $v$ in $Q_k$ to $dist$.

\begin{algorithm} 
\caption{SSSP algorithm based on a relaxed priority queue. }
\label{alg:1}

    \KwData{Graph $G=(V,E)$, source vertex $s$. \\ Initially empty relaxed priority queue $Q_k$. \\ Array $dist[n]$ for tentative distances. }
    \For {each vertex $v \in V$} { $dist[v] \gets+\infty$ }
    $dist[s] \gets 0$ \\ 
    $Q_k.push(s, 0)$ \\

    \While{$!Q_k.empty()$}
    {
        $(v,curDist)\gets Q_k.pop()$ \\ 
       
        \If {$curDist > dist[v]$} { 
            \Continue\ // $curDist$ is outdated 
        }
		\For {$u: (v,u) \in E$}
        {
            $e \gets (v, u)$\\
        	\If{$dist[u] > curDist + w(e)$}
        	{ 
            	$dist[u] \gets curDist + w(e)$ \\
            	{// We assume that we can check whether $v$ is in \\ //
            	$Q_k$, this can be implemented via maintaining \\ //
            	the corresponding flag for each vertex.} \\ 
            	\If {$u \in Q_k$}
            	{
            	    $Q_k.DecreaseKey(u, dist[u])$
            	}
            	\Else
            	{
            	    $Q_k.Add(u, dist[u])$
            	}
            }
        }
    }
\end{algorithm}
\todo{"dist" instead of "tent"}
\todo{... with an additional marker field in vertices}

\paragraph{Analysis.} 
We will prove the following statement, which upper bounds the extra work incurred by the relaxed scheduler:

\begin{theorem}
The number of $Q_k.pop()$ operations performed by Algorithm \ref{alg:1}
is $O(k^2d_{\max}/w_{\min})+n$.
\end{theorem}

\begin{proof}
Our analysis will follow the general pattern of $\Delta$-stepping analysis. 
We will partition the vertex set $V$ into buckets, based on distance: 
vertex $v$ belongs to bucket $B_i$ iff $d(v) \in [iw_{min},(i+1)w_{min})$. Let $t=d_{\max}/w_{\min}$ be the total number of buckets we need (for simplicity we assume that $d_{\max}/w_{\min}$ is an integer). 

Observe that because of the way we defined buckets, we have the following property, which we will call \emph{the bucket property} : for any vertex $v \in V$, no shortest path from $s$ to $v$ contains vertices which belong to the same bucket. 

We say that Algorithm \ref{alg:1} processes vertex $v$ \emph{at the correct distance} if $Q_k.pop()$ returns $(v, d(v))$, this means that $dist[v]=d(v)$ at this point and we relax outgoing edges of $v$. (See Algorithm~\ref{alg:1} for clarification.)

We fix $i<t$ and look at what happens 
when Algorithm \ref{alg:1} processes all vertices in the buckets 
$B_0, B_1, ..., B_i$ at the correct distance. 
Because of the bucket property, we get that $d(u)=dist[u]$ for every $u\in B_{i+1}$, and the vertices from bucket $B_{i+1}$ are either ready to be processed at the correct distance, or are already processed at the correct distance. To avoid the second case, we also assume that 
if $Q_k.pop()$ returns $(u, d(u))$, where $u \in B_{i+1}$ and not all vertices in the buckets $B_0, B_1, ..., B_i$ are processed at the correct distance, then this $Q_k.pop()$ operation still counts towards the total number of $Q_k.pop()$ operations,
but it does not actually remove the task and does not perform edge relaxations, even though $u$ is ready to be processed at the correct distance. This assumption only increases the total number of $Q_k.pop()$ operations, so to prove the claim it suffices to derive an upper bound for this pessimistic case.

Once the algorithm processes the vertices in buckets $B_0, B_1, ..., B_i$
at the correct distances, we know that the only vertices with tentative distance less than $(i+2)w_{min}$ 
are the vertices in the bucket $B_{i+1}$. 
(Note that this statement would not hold if we didn't have the $DecreaseKey$ operation: if we insert multiple copies of vertices in $Q_k$ with different distances, as in some versions of Dijkstra, there might exist outdated copies of vertex $u \in B_j, j<i$, even though $u$ was already processed at the correct distance.)
This means that, at this point, the top $|B_{i+1}|$ vertices (vertices with the smallest distance estimates) belong to $B_{i+1}$.

Next, we bound how many $Q_k.pop()$ operations are needed to process the vertices in $B_{i+1}$, after all vertices in the buckets $B_0, B_1, ..., B_i$
are processed.
If $|B_{i+1}|>k$, using the rank property, we have that the first ($|B_{i+1}|-k$) $Q.pop()$ operations process vertices in $B_{i+1}$. 
If $|B_{i+1}| \le k$, we know that it will take at most $k^2$ $Q_k.pop()$ operations to process all vertices in $B_{i+1}$, since, because by the fairness bound, the number of $Q_k.pop()$ operations to return the top vertex (the one with the smallest tentative distance) is at most $k$, and we showed that the top vertex belongs to $B_{i+1}$ until all vertices in $B_{i+1}$ are processed. By combining these two cases, we get that the  number of $Q_k.pop()$ operations to process vertices in $B_{i+1}$ at the correct distance is at most $|B_{i+1}|+k^2$. 

Thus the  number of $Q_k.pop()$ operations performed by Algorithm \ref{alg:1} in total is at most:
\begin{equation}
\sum_{i=0}^{t} (k^2+|B_i|)=n+O(k^2d_{\max}/w_{\min}),
\end{equation}

\noindent as claimed. 

\end{proof}

\paragraph{Discussion.} A clear limitation is that the bound depends on the maximum distance $d_{\max}$, and on the minimum weight $w_{\min}$. Hence, this bound would be relevant only for low-diameter graphs with bounded edge weights. We note however this case is relatively common: for instance,~\cite{dhulipala17julienne} considers weighted graph models of low diameter, where weights are chosen in the interval $[1, \log n)$. These assumptions appear to hold in for many publicly available weighted graphs~\cite{snapnets}. 
Further, our argument assumes a relaxed scheduler supporting $DecreaseKey$ operations. This operation can be supported by schedulers such as the SprayList~\cite{SprayList} or MultiQueues~\cite{MQ, AKLN17} where elements are hashed consistently into the priority queues.

\section{Experiments}

\begin{figure*}
    \includegraphics[width=0.9\textwidth]{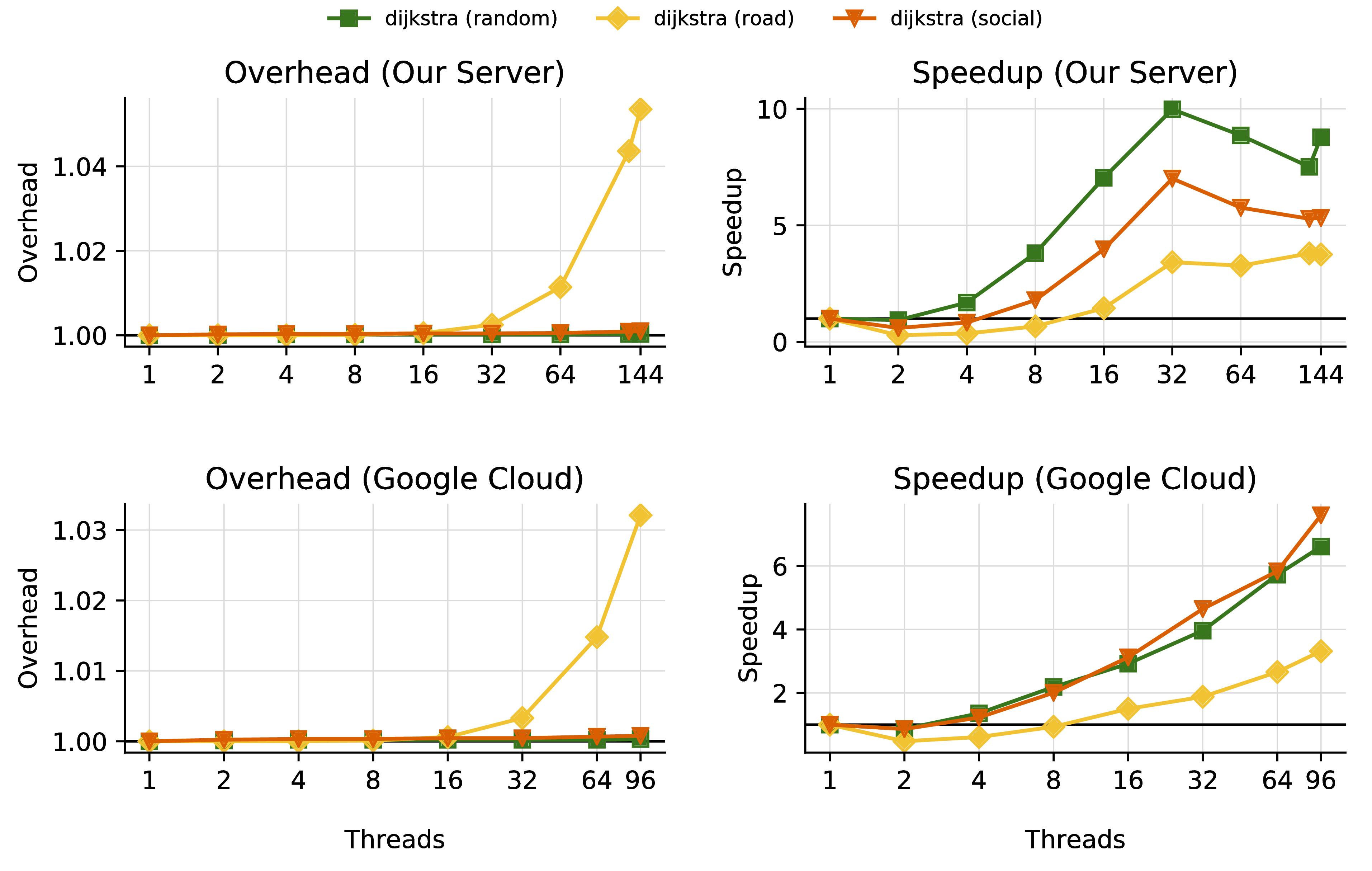}
    \caption{Overheads (left) and speedups (right) for parallel SSSP Dijkstra's algorithm executed via a MultiQueue relaxed scheduler on random, road network, and social network graphs. The overhead is measured as the ratio between the number of tasks executed via a relaxed scheduler versus an exact one.}
    \label{figure:experiment_dijkstra}
\end{figure*}

\begin{figure*}
    \includegraphics[width=0.8\textwidth]{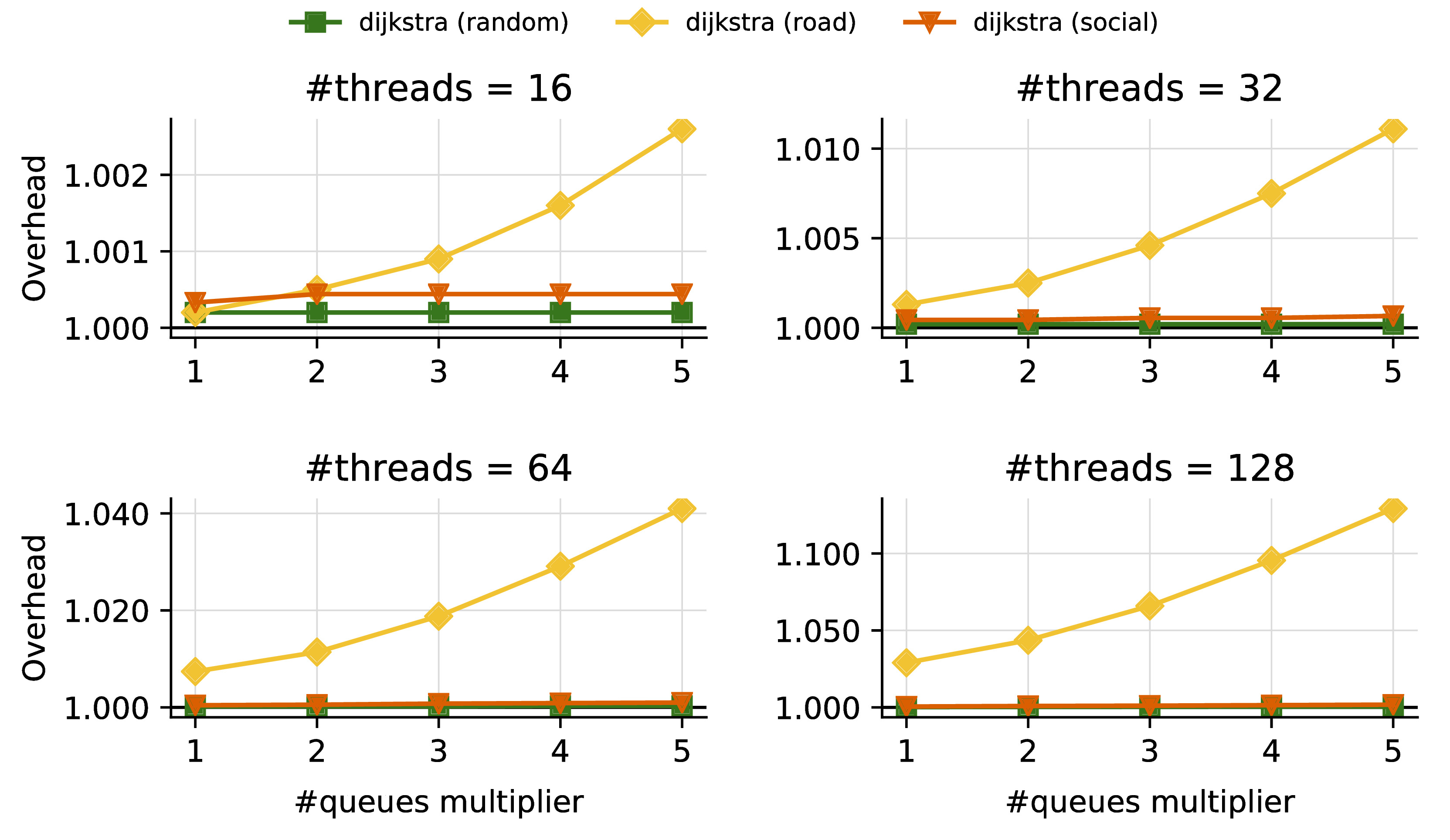}
    \caption{Relaxation overheads versus relaxation factor/queue multiplier for parallel SSSP Dijkstra's algorithm. The number of queues is the multiplier ($x$ axis) times the number of threads, and is proportional to the average relaxation  factor of the queue~\cite{AKLN17}. }
    \label{figure:dijkstra_queue_numbers}
\end{figure*}

We implemented the parallel SSSP Dijkstra's algorithm described in Section~\ref{section:sssp} using an instance of the MultiQueue relaxed priority scheduler~\cite{MQ, AKLN17}. In the classic sequential algorithm nodes are processed sequentially, while in this parallel version a node can be processed several times due to out-of-order execution. 
In our experiments, we are interested in the total number of tasks processed by the concurrent variant, in order to examine the overhead of relaxation in concurrent executions. In addition, we also measure execution times for increasing number of threads. \emph{Overhead} is measured as the average number of tasks executed in a concurrent execution divided by the number of tasks executed in a sequential execution using an exact scheduler. 

\paragraph{Sample graphs.}
We use the following list of graphs in our experiments:
\begin{itemize}
    \item Random undirected graph with $1$ million nodes and $10$ million edges, with uniform random weights between $0$ and $100$ (\textbf{random});
    \item USA road network graph with physical distances as edge lengths; $\sim24$ million nodes and $\sim58$ million edges (\textbf{road}) \cite{demetrescu2009shortest}; 
    \item LiveJournal social network friendship graph; $\sim5$ million nodes and $\sim69$ million edges, with uniform random weights between $0$ and $100$ (\textbf{social}) \cite{snapnets}.
\end{itemize}

\paragraph{Platforms.}
We evaluated the experiment on a server with 4 Intel Xeon Gold 6150 (Skylake) sockets. Each socket has 18 2.70 GHz cores, each of which multiplexes 2 hardware threads, for a total of 144 hardware threads. 
In addition, we ran the experiment on a Google Cloud Platform VM supporting to 96 hardware threads.

\paragraph{Experimental results.} The experimental results are summarized in Figure~\ref{figure:experiment_dijkstra}. On the left column, notice that, on both machines, the overheads of relaxation are almost negligible: for the random graph and the social network, the overheads are almost $1\%$ at all thread counts, what practically means the absence of extra work. (Recall that the number of queues is always $2\times$ the number of threads, so the relaxation factor increases with the thread count.) 

The road network incurs higher overheads ($5\%$ at 144 threads / 288 queues). This can be explained by the higher diameter of the graph ($6261$, versus $16$ for the LiveJournal and $6$ for the random graphs), and by the higher variance in edge costs for the road network. 
In terms of speedup (right), our implementation scales well for 1-2 sockets on our local server, after which NUMA effects become prevalent. 
NUMA effects are less prevalent on the Google Cloud machine, but the maximum speedup is also  more limited ($<7\times$ instead of $~10\times$). 
In Figure~\ref{figure:dijkstra_queue_numbers}, we examine the relaxation overhead (in terms of the amount of extra tasks executed) versus the relaxation factor. While we cannot control the relaxation factor exactly, we know that the average value of this factor is proportional to the number of queues allocated, which is the number of threads (fixed for each sub-plot) times the multiplier for the number of queues (the $x$ axis)~\cite{AKLN17}. We notice that these overheads are only non-negligible for the road network graph. On the one hand, this suggests that our worst-case analysis is not tight, but can also be interpreted as showing that the overheads of relaxation do become apparent on dense, high-diameter graphs such as road networks.

\section{Conclusion}

We have provided the first efficiency bounds for parallel implementations of SSSP and Delaunay mesh triangulation under relaxed schedulers. 
In a nutshell, our results show that, for some inputs and under analytic assumptions, the overheads of parallelizing these algorithms via relaxed schedulers can be negligible. Our findings complement empirical results showing similar trends in the context of high-performance relaxed schedulers~\cite{Nguyen13, lenharth2015priority}. 
While our analysis was specialized to these algorithms, we believe that our techniques can be generalized to other iterative algorithms, which we leave as future work.

\section*{Acknowledgments}

We would like to thank Ekaterina Goltsova, Charles E. Leiserson, Tao B. Schardl, and Matthew Kilgore for useful discussions in the incipent stages of this work, and Justin Kopinsky for careful proofreading and insightful suggestions on an earlier draft. 

Giorgi Nadiradze was supported by the Swiss National Fund Ambizione Project PZ00P2 161375. Dan Alistarh was supported by European Research Council funding award PR1042ERC01.

\bibliographystyle{plain}

\end{document}